\documentclass[11pt]{article}
\usepackage{titlesec}
\titleformat*{\section}{\centering}
\usepackage{amsmath,amssymb,amsfonts,amsthm,bbm}
\usepackage{epic,eepic,epsfig,longtable}
\usepackage{multirow,verbatim}
\usepackage{array}
\usepackage{graphicx}
\usepackage{floatrow}
\usepackage{float}
\usepackage{subcaption}
\usepackage{paralist}
\usepackage{latexsym}
\usepackage{comment}
\usepackage{epsfig}
\usepackage{setspace}
\usepackage{rotating}
\usepackage{CJK}
\usepackage{color}
\usepackage{mathtools}

\usepackage{caption}
\usepackage[authoryear,round]{natbib}
\makeatletter
\def\singlespace{\def\baselinestretch{1}\@normalsize}

\makeatletter
\def\singlespace{\def\baselinestretch{1}\@normalsize}

\numberwithin{equation}{section}

\renewcommand{\hat}{\widehat}

\renewcommand{\hat}{\widehat}

\newcommand{\bfsym}[1]{\ensuremath{\boldsymbol{#1}}}

\def\1{\bfsym{1}}	

\def\newpage{\vfill\eject}

\def\today{\ifcase\month\or
	January\or February\or March\or April\or May\or June\or
	July\or August\or September\or October\or November\or December\fi
	\space\number\day, \number\year}

\newdimen\biblioindent    \biblioindent=30pt

 at 8truept

\newcommand{\beq}{\begin{equation}}
\newcommand{\eeq}{\end{equation}}
\newcommand{\beqn}{\begin{eqnarray}}
\newcommand{\eeqn}{\end{eqnarray}}
\newcommand{\beqnn}{\begin{eqnarray*}}
	\newcommand{\eeqnn}{\end{eqnarray*}}

\allowdisplaybreaks
\usepackage{verbatim}
\def\tilde{\widetilde}

\def\[{\left [}  \def\]{\right ]} \def\({\left (}  \def\){\right )}
 
\def\hat{\widehat}

\newtheorem{assumption}{Assumption}
\newtheorem{theorem}{Theorem}
\newtheorem{lemma}{Lemma}

\newtheorem{proposition}{Proposition}
\theoremstyle{definition}
\newtheorem{definition}{Definition}

\title{A NEW VOLATILITY MODEL:  GQARCH-IT\^{O} MODEL}
\author{HUILING YUAN,$^a$
	YONG ZHOU,$^{c\&d}$
     LU XU$^e$ 
     YULEI SUN$^b$AND
	XIANGYU CUI$^b$\footnote {Corresponding author: Xiangyu Cui. Address: 777 Guoding Rd., Shanghai,200433, P. R. China. Tel: 86-21-6590 4311.   Fax: 86-21-6590 1079. E-mail: cui.xiangyu@mail.shufe.edu.cn.} \\
	\\$^a$ \small{School of Data Science,
	City University of Hongkong, Hongkong, China}\\
   $^b$ \small{School of Statistics and Management, 
	Shanghai University of Finance and Economics, Shanghai, China}\\
	$^c$ \small{Institute of Statistics and Interdisciplinary Sciences and School of Statistics,}\\
	\small{Faculty of Economics and Management, 
	East China Normal University, Shanghai, China} \\
	$^d$ \small{Academy of Mathematics and Systems Sciences, 
	Chinese Academy of Sciences, Beijing, China} \\
	$^e$ \small{Shenyin \& Wanguo Fortune Investment Co., Ltd, Shanghai, China}\\
}
\date{}
\begin{document}
	\maketitle
	
%	\begin{abstract}
		
	Volatility asymmetry is a hot topic in high-frequency financial market. In this paper, we propose a new econometric model, which could describe volatility asymmetry based on high-frequency historical data and low-frequency historical data. After providing the quasi-maximum likelihood estimators for the parameters, we establish their asymptotic properties. We also conduct a series of simulation studies to check the finite sample performance and volatility forecasting performance of the proposed methodologies. And an empirical application is demonstrated that the new model has stronger volatility prediction power than GARCH-It\^{o} model in the literature.\\
		
%	\end{abstract}


	\noindent \textbf{MOS subject classification:} 62M10, 62M20, 62F12
	%%%%%%%%
	%%%%%%%%
	%%%%%%%%
	
	%----------------------------------------------------------------------------------------
	%	SECTION 1 (Introduction)
	%----------------------------------------------------------------------------------------
	
	\renewcommand{\baselinestretch}{1.5}
	\baselineskip=22pt
	
	\section{INTRODUCTION} 
	\label{SEC-1}
Volatility measures plays a crucial role in modern financial markets. The wide information source of modeling the volatility is the historical data of the security, which can be further
divided into low-frequency data and high-frequency data.
High-frequency data and Low-frequency data are two different time scale in financial market. High-frequency data are observed at intra-day for financial assets, while low-frequency data are referred for financial assets at daily or longer time horizons. There are many important models such as Generalized autoregressive conditional heteroskedasticity
(GARCH) models  \citep{Bollers:1986} used in volatility analysis of low-frequency data. The standard GARCH models describe the squared daily log returns as the conditional volatilities. On the other hand, there are also several well-performing realized volatility estimators, for example, two-time scale realized volatility(TSRV) \citep{Zh:2005}, multi-scale realized volatility(MSRV) \citep{Zh:2006}, kernel realized volatility(KRV) \citep{B:2009}, pre-averaging realized volatility(PRV) \citep{J:2009} and quasi-maximum likelihood estimator(QMLE) \citep{X:2010}. However, these models were developed for high-frequency data and low-frequency data quite independently. In fact, high-frequency data and low-frequency data must be inter-related at different time scales, due to just these different time scales. Fortunately, there are some attempts to bridge the gap between high-frequency data and low-frequency data. \cite{Wang:2002} studies the statistical relationship between the GARCH and diffusion model. \cite{Hansen:2012} studies volatilities analysis by combining the realized GARCH model and the high-frequency volatility model.  \cite{KW:2016} proposed GARCH-It$\hat{o}$ model for merged low-frequency data and high-frequency data.

GARCH-It$\hat{o}$ model is a unified model for both high-frequency data and low-frequency data, which is a continuous-time It$\hat{o}$ process at high-frequency data points, and a GARCH model at integer time points. The parameters estimators based on GARCH-It$\hat{o}$ model have better performances than the estimators using only low-frequency data by asymptotic theory and simulation study.

However, GARCH-It$\hat{o}$ model has not introduced the volatility asymmetry. Volatility asymmetry is an important phenomenon in financial market. There are two explanation for volatility asymmetry: leverage effect and volatility feedback effect. \cite{Black:1976} and \cite{Christie:1982} first gave the description of leverage effect: when the asset prices is declining, the companies leverage (dept-to-equity ratio)become larger, so the stock becomes riskier since its volatility is increasing. Therefore, leverage effect implies a negative correlation structure between the analysed asset return and its volatility changes. On the other hand, \cite{French et al:(1987)} proposed the volatility feedback effect: if volatility is priced, and anticipated increase in volatility would raise the required rate of return, in turn necessitating an immediate stock-price decline to allow for higher future returns.

In recent years, many scholars have studied volatility asymmetry based high-frequency financial data. \cite{Bouchaud et al.:2001} proposed the peak effect at the instantaneous correlation between return and volatility over fairly small time intervals. By an application of high-frequency five-minute S\&P 500 futures, \cite{BLT:2006} found that there exists significantly negative correlation for several days between the absolute high-frequency returns and the current and past returns, and low correlations between the volatility and lagged return. \cite{BO:2009} obtained a highly accurate discrete-time daily stochastic volatility model that distinguishes between the jump and continuous-time components of price movements. \cite{WM:2014} proposed the new nonparametric estimators of leverage effect based on the stochastic volatility model. \cite{KalninaXiu:2017} provided the integrated leverage effect estimator, and gave its the statistical properties. \cite{Curato:2019} presented the non-parameter estimator of leverage effect via Fourier transformation. \cite{Bibinger:2019} explored the non-continuous leverage effect in 320 NASDAQ corporations.

As we all known, GARCH-It$\hat{o}$ model have better statistical performance than other current volatility models. It is natural to extend  GARCH-It$\hat{o}$ model to a unified model describing the volatility asymmetry. Many empirical analysis told us that the complex models would sometimes perform worse properties, and time-consuming, therefore, we try to look for a volatility model with minor modifications for GARCH-It$\hat{o}$ model. Fortunately, \cite{sentana:1995} proposed Quadratic ARCH(QARCH) and Generalized QARCH(GQARCH) model, which can be integrated in economic models and provides a very simple way of calibrating and testing for dynamic asymmetries for some financial times. Inspired by this, we expand GARCH-It$\hat{o}$ model so that features of financial data at both frequencies can be better captured as follows. First, volatility asymmetry that are described in empirical studies are allowed. Second, we explore that volatility forecasting performances for $\frac{1}{j}, j=1,2,3,4,5,6$ of daily volatility. We name the proposed model as the GQARH-It$\hat{o}$ model. The key feature of the proposed model is that its conditional volatility has integrated volatility and asymmetry as innovations.

The paper is organized as follows. Section 2 introduces the GQARCH-It$\hat{o}$ model. Section 3 introduces quasi-likelihood estimation methods and investigates their asymptotic behaviors. Section 4 conducts the simulation studies to check the finite sample performance and different volatility forecasting performances for the proposed model. Section 5 carries out an empirical analysis with CRPS total market index  to demonstrate the advantage of the proposed model in volatility forecasting. We collect the proofs in the Appendix.

\section{GQARCH-IT\^{O} MODEL}

\subsection{GQARCH(1,1) model at discrete-time}

In order to capture dynamic asymmetric that GARCH model rules out, \cite{sentana:1995} proposed GQARCH model, which allows an asymmetric effect on the conditional variance, and GQARCH(1,1) model structure is as follows, 
\begin{align*}
X_{t}-X_{t-1}=&\mu+\xi_{t},\\
\xi_{t} = & \sigma_{t} \varepsilon_{t},\\
\sigma_{t}^{2} = & \kappa+\sigma_{t-1}^{2}+\varphi \xi_{t-1}^2+\phi \xi_{t-1} ,
\end{align*}
%$X_{t}-X_{t-1}=\mu+\xi_{t},  t=1,2,\cdots, n,$
where $X_{t}$ is the ture log price at integrated time $t=1,2,\cdots, n$,  $\sigma_{t}$ is volatility,  $\varepsilon_{t}$ is i.i.d  varibles and random errors $\xi_{t}$ satisfy  $E\left[\xi_{t} \bigg|\mathcal{F}_{t-1}^{LF}\right]=0$ a.s.,  $\mathcal{F}_{t}^{LF}=\sigma(X_{t},X_{t-1},\cdots)$. Therefore,  their conditional variances obey,
\begin{equation}\label{eq_cv}
E\left[\xi_{t} ^2 \bigg|\mathcal{F}_{t-1}^{LF} \right]=\kappa +E\left[\xi_{t-1} ^2 \bigg|\mathcal{F}_{t-2}^{LF}\right]+\varphi \xi_{t-1} ^2+\phi \xi_{t-1}.
\end{equation}
Equation (\ref{eq_cv}) is different from the conditional variances of GARCH model, this is because GQARCH model incorporate the term $\xi_{t-1}$.
%$\displaystyle\int_{[t]}^{t}\sigma_{s}d B_{s}$. 

\subsection{GQARCH-It\^{o} model}
In this subsection, we try to provide a volatility model that describes the  volatility asymmetry by embedding a standard GQARCH(1, 1) model into an It\^{o} process with an instantaneous volatility as follows. Noted that $\mathbb{R}_{+}= [0, \infty]$ and $\mathbb{N}$ is the set of all non-negative integers.
\begin{definition}\label{def_1}
	We call a log security price $X_t$, $t\in[0,+\infty)$, to follow a unified GQARCH-It\^{o} model, if it satisfies
	\begin{equation}\label{eq:garch-ito-iv}
	\begin{array}{rl}
	dX_{t}=&\mu dt+\sigma_{t}dB_{t},\\
	\sigma_{t}^{2}=&\sigma_{[t]}^{2}+(t-[t])\left\{\omega+(\gamma-1)\sigma_{[t]}^{2}\right\}+\beta\left(\displaystyle\int_{[t]}^{t}\sigma_{s}d B_{s}\right)^{2}+\alpha\displaystyle\int_{[t]}^{t}\sigma_{s}d B_{s},
	\end{array}
	\end{equation}
	where $\mu$ is a drift, $B_{t}$ is a standard Brownian motion with respect to a filtration $\mathcal{F}_{t}$, $\sigma_t^2$ is the instantaneous volatility process adapted to $\mathcal{F}_t$, $[t]$ denotes the integer part of $t$, $\theta=\left(\omega,\gamma,\beta,\alpha\right)$.
\end{definition}

According to Definition \ref{def_1}, GQARCH-It\^{o} model is a continuous It\^{o} process defined at all times  $t\in \mathbb{R}_+$. When it is restricted to integer times $t\in \mathbb{N}$,  the conditional variance of daily return $X_{t}-X_{t-1}$  follows a GQARCH(1,1) structure,
\begin{align*}
E\left[Z_{t} ^2 \bigg|\mathcal{F}_{t-1}^{LF} \right]=&\omega_{1}^{g} +\gamma E\left[Z_{t-1} ^2 \bigg|\mathcal{F}_{t-2}^{LF}\right]+\beta_{1}^{g} Z_{t-1} ^2+\alpha_{1}^{g}Z_{t-1},\\
Z_{t}=&\displaystyle\int_{t-1}^{t}\sigma_{s}d B_{s}, t\in\mathbb{N},
\end{align*}
where $\mathcal{F}_{t}^{LF}=\sigma(X_{t},X_{t-1},\cdots)$, $	\omega_{1}^{g}=\beta^{-1}\left(e^{\beta}-1\right)\omega, \beta_{1}^{g}=\beta^{-1}(\gamma-1)\left(e^{\beta}-1-\beta\right)+e^{\beta}-1,
$ and $\alpha_{1}^{g}=\alpha\left(\beta^{-2}(\gamma-1)\left(e^{\beta}-1-\beta\right)+\beta^{-1}\left(e^{\beta}-1\right)\right)$ (see: Proposition \ref{prop:1}).
 Therefore, our proposed GQARCH-It\^{o} model could capture the asymmetries based on low-frequency historical data and high-frequency historical data. The asymmetry explanation is the same  as GQARCH model in \cite{sentana:1995}' paper.

\subsection{Integrated volatility for GQARCH-It\^{o} model}
In general, the trade period is about six and half hours in financial market, the financial sectors often focus on the several hours volatility prediction rather than daily volatility. Therefore, we would study the $\frac{1}{j}, j=1,2,3,4,5,6$ of daily volatility obtained from the GQARCH-It\^{o} model over consecutive integers.

\begin{proposition}\label{prop:1}
	(a) Under GQARCH-It\^{o} model,  we have, for any $k,n \in \mathbb{N}$ and $0< \beta <1$,
	\begin{align*}
	R(k)\equiv & \int_{n-\frac{1}{j}}^{n}\frac{(n-t)^{k}}{k!}\sigma_{t}^{2}dt\\
	= & \frac{ \omega +(\gamma +j(k+2)-1) \sigma_{n-\frac{1}{j}}^{2}}{(k+2)!}\left(\frac{1}{j}\right)^{k+2}+\beta R(k+1)+2 \beta
	\int_{n-\frac{1}{j}}^{n}\frac{(n-t)^{k+1}}{(k+1)!}\int_{n-\frac{1}{j}}^{t}\sigma_{s}dB_{s}\sigma_{t}dB_{t}\\
	+&\alpha \int_{n-\frac{1}{j}}^{n}\frac{(n-t)^{k+1}}{(k+1)!}\sigma_{t}dB_{t}.
	\end{align*}
	In particular,
	\begin{align*}
	&\int_{n-\frac{1}{j}}^{n}\sigma_{t}^{2}dt\\
	=&\sum_{k=0}^{\infty} \frac{\beta^{k} \left[\omega +(\gamma +j(k+2)-1) \sigma_{n-\frac{1}{j}}^{2}\right]}{(k+2)!}\left(\frac{1}{j}\right)^{k+2}+ \sum_{k=0}^{\infty}2
	\int_{n-\frac{1}{j}}^{n}\frac{\left[\beta(n-t)\right]^{k+1}}{(k+1)!}\int_{n-\frac{1}{j}}^{t}\sigma_{s}dB_{s}\sigma_{t}dB_{t}\\
	+&\sum_{k=0}^{\infty}\int_{n-\frac{1}{j}}^{n}\frac{\alpha\left(n-t\right)^{k+1}}{(k+1)!}\beta^{k}\sigma_{t}dB_{t}\\
	=& g_{n}(\theta)+D_{n},
	\end{align*}
	where
	\begin{align}\label{eq:h_nj}
	g_{n}(\theta)=&\omega_{j}^{g}+\frac{\gamma+j-1}{j} g_{n-\frac{1}{j}}(\theta)+\beta_{j}^{g}Z_{n-\frac{1}{j}}^{2}+\alpha_{j}^{g}Z_{n-\frac{1}{j}},
	\end{align}
	the parameters are given as
	\begin{align*}
	\omega_{j}^{g}=&\beta^{-1}\left(e^{\frac{\beta}{j}}-1\right)\frac{\omega}{j},\quad \beta_{j}^{g}=\beta^{-1}(\gamma-1)\left(e^{\frac{\beta}{j}}-1-\frac{\beta}{j}\right)+e^{\frac{\beta}{j}}-1,\\
	\alpha_{j}^{g}=&\alpha\left(\beta^{-2}(\gamma-1)\left(e^{\frac{\beta}{j}}-1-\frac{\beta}{j}\right)+\beta^{-1}\left(e^{\frac{\beta}{j}}-1\right)\right),
	\end{align*}
	and
	\begin{align*}
	D_{n}=2\int_{n-\frac{1}{j}}^{n}\left(e^{(n-t)\beta}-1\right)\int_{n-\frac{1}{j}}^{t}\sigma_{s}dB_{s}\sigma_{t}dB_{t}+\frac{\alpha}{\beta}\int_{n-\frac{1}{j}}^{n}(e^{\left(n-t\right)\beta}-1)\sigma_{t}dB_{t},
	\end{align*}
	is a martingale difference.
	
	(b) For any $k,n \in \mathbb{N}$ and $0< \beta <1$, we have
	\begin{align*}
	&E\left[\int_{n-\frac{1}{j}}^{n}\frac{(n-t)^{k}}{k!}\sigma_{t}^{2}dt \bigg|\mathcal{F}_{n-1}\right]\\
	=&\frac{ \omega +(\gamma +j(k+2)-1) \sigma_{n-\frac{1}{j}}^{2}}{(k+2)!}\left(\frac{1}{j}\right)^{k+2}+\beta
	E\left[\int_{n-\frac{1}{j}}^{n}\frac{(n-t)^{k+1}}{(k+1)!}\sigma_{t}^{2}dt\bigg|\mathcal{F}_{n-1}\right].
	\end{align*}
	In particular,
	\begin{align*}
	E\left[\int_{n-\frac{1}{j}}^{n}\sigma_{t}^{2}dt\bigg|\mathcal{F}_{n-1}\right]=g_{n}(\theta).
	\end{align*}
\end{proposition}

When $j=1$, Proposition 1 denotes the daily volatility $\int_{n-1}^{n}\sigma_{t}^{2}dt$ consist of $g_{n}(\theta)$ and a martingale difference $D_{n}$.
Where
\begin{align}\label{g_n}
g_{n}(\theta)=&\omega_{1}^{g}+\gamma g_{n-1}(\theta)+\beta_{1}^{g}Z_{n-1}^{2}+\alpha_{1}^{g}Z_{n-1},
\end{align}
and
\begin{align}
D_{n}=2\int_{n-1}^{n}\left(e^{(n-t)\beta}-1\right)\int_{n-1}^{t}\sigma_{s}dB_{s}\sigma_{t}dB_{t}+\frac{\alpha}{\beta}\int_{n-1}^{n}(e^{\left(n-t\right)\beta}-1)\sigma_{t}dB_{t},
\end{align}
According to (\ref{g_n}), $g_{n}(\theta)$ includes the term $\alpha_{1}^{g}Z_{n-1}$, which denotes the volatility asymmetry.

\section{PARAMETER ESTIMATION FOR GQARCH-IT\^{O} MODEL}

\subsection{Quasi-maximum likelihood estimation}
The underlying log price process is assumed to obey the GQARCH-It\^{o} model as described in Definition \ref{def_1}. Let $n$ be the total number of low-frequency observations and $m_{i}$ be the total number of high-frequency observations duiring the $i$th low-frequency period, for example, the $i$th day. Furthermore, the low-frequency historical data are observed true log prices at integer times, namely $X_{i}, i= 0, 1, 2, \cdots, n$, and the high-frequency historical data are
observed log prices at time points between integer times, that is, $t_{i,j}$, $j=0,1,\cdots,m_i$, denote the high-frequency time points
during the $i$-th period satisfying $i-1=t_{i,0}<t_{i,1}<\ldots<t_{i,m_{i}}<t_{i,m_{i}}=t_{i+1,0}=i$. As we all known, the true high-frequency log prices are not observable, and the observed high-frequency log prices are contaminated by the market micro-structure noise. In this regard, we assume that observed high-frequency log prices $Y_{t_{i,j}}$ obey the simple additive noise model,
\begin{align}\label{eq:y}
Y_{t_{i,j}}=X_{t_{i,j}}+\epsilon_{t_{i,j}},
\end{align}
where $\epsilon_{t_{i,j}}$ is micro-structure noise independent of the process of $X_{t_{i,j}}$,
and for each $i$, $\epsilon_{t_{i,j}},j=1,\cdots,m_{i}$, are independent and identically distributed (i.i.d.) with mean zero and variance $a^{2}$.

Similar to \cite{KW:2016}, we propose the quasi-likelihood function $\widetilde{L}_{n,m}^{GH}$ for GQARCH-It\^{o} model as follows,
\begin{align*}
\widetilde{L}_{n,m}^{GH}(\theta)=&-\frac{1}{2n}\sum_{i=1}^{n}\left(\log g_{i}(\theta)+\frac{RV_{i}}{g_{i}(\theta)}\right),
\end{align*}
where $g_{i}(\theta)$ has the structure of (\ref{g_n}), and the realized volatility $RV_{i}$ is computed using $m_{i}$ high-frequency historical data during the $i$-th period and is treated as an ``observation''. $RV_{i}$ could be estimated by two-time scale realized volatility (TSRV) (\citealp{Zh:2005}), multi-scale realized volatility (MSRV) (\citealp{Zh:2006}), kernel realized volatility (KRV) (\citealp{B:2009}), pre-averaging realized volatility (PRV) (\citealp{J:2009}) and quasi-maximum likelihood estimator (QMLE) (\citealp{X:2010}) among others.

We maximize the quasi-likelihood function $\widetilde{L}_{n,m}^{GH}(\theta)$ over parameters'  space $\theta$ and denote the maximizer as $\tilde\theta^{GH}$, that is,
\begin{align*}
\tilde\theta^{GH}=\arg\max \limits_{\theta\in\Theta}\widetilde{L}_{n,m}^{GH}(\theta).
\end{align*}
$\tilde\theta^{GH}= (\omega^{GH}, \beta^{GH}, \gamma^{GH}, \alpha^{GH})$ are the quasi-maximum likelihood estimators of $\theta_0 = (\omega_{0}, \beta_{0}, \gamma_0, \alpha_{0})$.

\subsection {Asymptotic theory of estimators}
In this subsection, we establish consistency and asymptotic distribution for the proposed estimators $\tilde\theta^{GH}= (\omega^{GH}, \beta^{GH}, \gamma^{GH}, \alpha^{GH})$.

First, we fix some notations. Given a random variable $X$ and $p \geq 1$, let $||X||_{L_p}= \{E [ |X|^{p}]\}^{1/p}$. For a matrix $A = (A_{i,j})_{i,j=1,\dots,k}$, and a vector $a = (a_1, \dots , a_k)$, define $||A||_{\max} = \max_{i,j} |A_{i,j}|$ and $||a||_{\max} = \max_{i} |a_i |$. Let $C$ be positive generic constants whose values are free of $\theta$, $n$ and $m_i$, and may change from appearance to appearance. Then, we give the following assumptions, under which the asymptotic theory is established.

\begin{assumption}\label{assumption}
	(a) Let
	\begin{align*}
	\Theta=&\{\theta=(\omega_{1},\beta_{1},\gamma,\alpha_{1} )~|~\omega_{l}<\omega_{1}<\omega_{u}, ~\beta_{l}<\beta_{1}<\beta_{u},~ \gamma_{l}<\gamma<\gamma_{u},
	~\alpha_{l}<\alpha_{1}<\alpha_{u},
	\end{align*}
	where $\omega_{l}$, $\omega_{u}$, $\beta_{l}$, $\beta_{u}$, $\gamma_{l}$, $\gamma_{u}$, $\alpha_{l}$, $\alpha_{u}$,are known constants.
	
	(b) $\left\{|D_{i}| ~|~ i \in \mathbb{N}\right\}$ is uniformly integrable.
	
	(c) One of the following conditions is satisfied.
	
	\quad \quad(c1) $\frac{E\left[Z_{i}^{4}|\mathcal{F}_{i-1}\right]}{g_{i}^{2}(\theta_{0})}\leq C$ a.s. for any $i\in \mathbb{N}$.
	
	\quad \quad(c2) There exists a positive constant $\delta$ such that $E\left[\left(\frac{Z_{i}^{2}}{g_{i}(\theta_{0})}\right)^{2+\delta}\right]\leq C$ for $i\in \mathbb{N}$.
	
	(d) $(D_{i},Z_{i}^{2})$ is a stationary ergodic process.
	
	(e) Let $m=\sum_{i=1}^{n}m_{i}/n$. We have $C_{1}m\leq m_{i} \leq C_{2}m$, $\sup\limits_{1\leq j\leq m_{i}}|t_{i,j}-t_{i,j-1}|=O(m^{-1})$ and $n^{2}m^{-1}\rightarrow 0$ as
	$m, n\rightarrow \infty$.
	
	(f) $\sup\limits_{i\in \mathbb{N}}\left\|RV_{i}-\int_{i-1}^{i}\sigma_{t}^{2}dt\right\|_{L_{1+\delta}}\leq C \cdot m^{-1/4}$ for some $\delta >0 $.
	
	(g) For any $i\in \mathbb{N}$, $E\left[RV_{i}|\mathcal{F}_{i-1}\right]\leq C\cdot E[\int_{i-1}^{i}\sigma_{t}^{2}dt|\mathcal{F}_{i-1}]+C$ a.s.
\end{assumption}

Comparing to the Assumption 1 in \cite{KW:2016}, we add one additional Assumption \ref{assumption}(a) on $\alpha_{l}<\alpha<\alpha_{u}$. Among Assumption \ref{assumption}, (a)-(d) are for the low-frequency part of the model, while  (e)-(g) are for the high-frequency part of the model. Similar to the explanation in \cite{KW:2016}, these assumptions are also all reasonable in this paper. The following theorem would give the consistency and convergence rate for $\tilde\theta^{GH}$.

\begin{theorem}
	(a) Under Assumption 1(a), (b), (d), (f)-(g), there is a unique maximizer of $L_{n}^{GH}(\theta)$ and as $m,n \rightarrow\infty$,  $\tilde{\theta}^{GH}\rightarrow
	\theta_{0}$ in probability, where
	$$L_{n}^{GH}(\theta)=-\frac{1}{2n}\sum_{i=1}^{n}\log(g_{i}(\theta))-\frac{1}{2n}\sum_{i=1}^{n}\frac{g_{i}(\theta_{0})}{g_{i}(\theta)}.$$
	
	(b) Under Assumption 1(a)-(d), (f)-(h), we have
	\begin{align*}
	\left\|\tilde{\theta}^{GH}-\theta_{0}\right\|_{\max} = O_{p}\left(m^{-1/4}+n^{-1/2}\right).
	\end{align*}
\end{theorem}

Theorem 1 shows that $\tilde{\theta}^{GH}$ has the same convergence rate as the parameter estimators in GARCH-It\^{o} model of \cite{KW:2016}. In other words, the asymmetry information has no significant effect on the converge rate of the parameter estimators. The following theorem shows the asymptotic normality of $\tilde{\theta}^{GH}$.

\begin{theorem}
	Under Assumption 1, we have as $m, n\rightarrow \infty$,
	\begin{align*}
	\sqrt{n}(\tilde{\theta}^{GH}-\theta_{0})\xrightarrow{d}N(0,B^{-1}A^{GH}B^{-1}),
	\end{align*}
	where
	\begin{align*}
	A^{GH}=&\frac{1}{4}E\left[\frac{\partial g_{1}(\theta)}{\partial\theta}\frac{\partial
		g_{1}(\theta)}{\partial\theta^{T}}\bigg|_{\theta=\theta_{0}}g_{1}^{-4}(\theta_{0})\right.\\
	&\left.\int_{0}^{1}\left[4(e^{\beta_{0}(1-t)}-1)^{2}(X_{t}-X_{0})^{2}
	+\frac{4\alpha_{0}}{\beta_{0}}
	(e^{\beta_{0}(1-t)}-1)^{2}(X_{t}-X_{0})+\frac{\alpha_{0}^{2}}{\beta_{0}^{2}}(e^{\beta_{0}(1-t)}-1)^{2}
	\right]\sigma_{t}^{2}dt\right],\\
	B=&\frac{1}{2}E\left[\frac{\partial g_{1}(\theta)}{\partial\theta}\frac{\partial g_{1}(\theta)}{\partial\theta^{T}}\bigg|_{\theta=\theta_{0}}g_{1}^{-2}(\theta_{0}) \right].
	\end{align*}
\end{theorem}
According to Theorem 2, we could see that $A^{GH}$ represent the influences of volatility asymmetry on the asymptotic variances of the parameter estimators.

\section{SIMULATION STUDY}
\subsection{Asymptotic Performance}
In this section, we generate the log prices $X_{t_{i,j}}, ~t_{i,j}= i-1+j/m,~ i=1, \cdots, n,~ j=1, \cdots, m$,  from the GQARCH-It\^{o} models with $\theta_{0}=(\omega_{0},\beta_{0},\gamma_{0},\alpha_{0})=(0.2,0.3,0.4,0.1)$. We further set $n=250, m=2160$, and $X_0=10$, which implies that $\sigma_{0}^{2}=0.667$. The high-frequency data $Y_{t_{i,j}}$ are obtained through model (\ref{eq:y}) by simulating $\epsilon_{t_{i,j}}$ from  $N(0, 0.001^2)$.  The multi-scale realized volatility estimator is used to estimate $RV_{i}$. All boxplots of estimators are displayed in Figure \ref{fig_boxplot}. We can see clearly that our proposed estimators have good statistical performance and the simulation confirms most of the theoretical findings in Section 3.

\begin{figure}[!htb]
	\centering
	\begin{minipage}[t]{5.6cm}
		\includegraphics[angle=0,width=5.5cm]{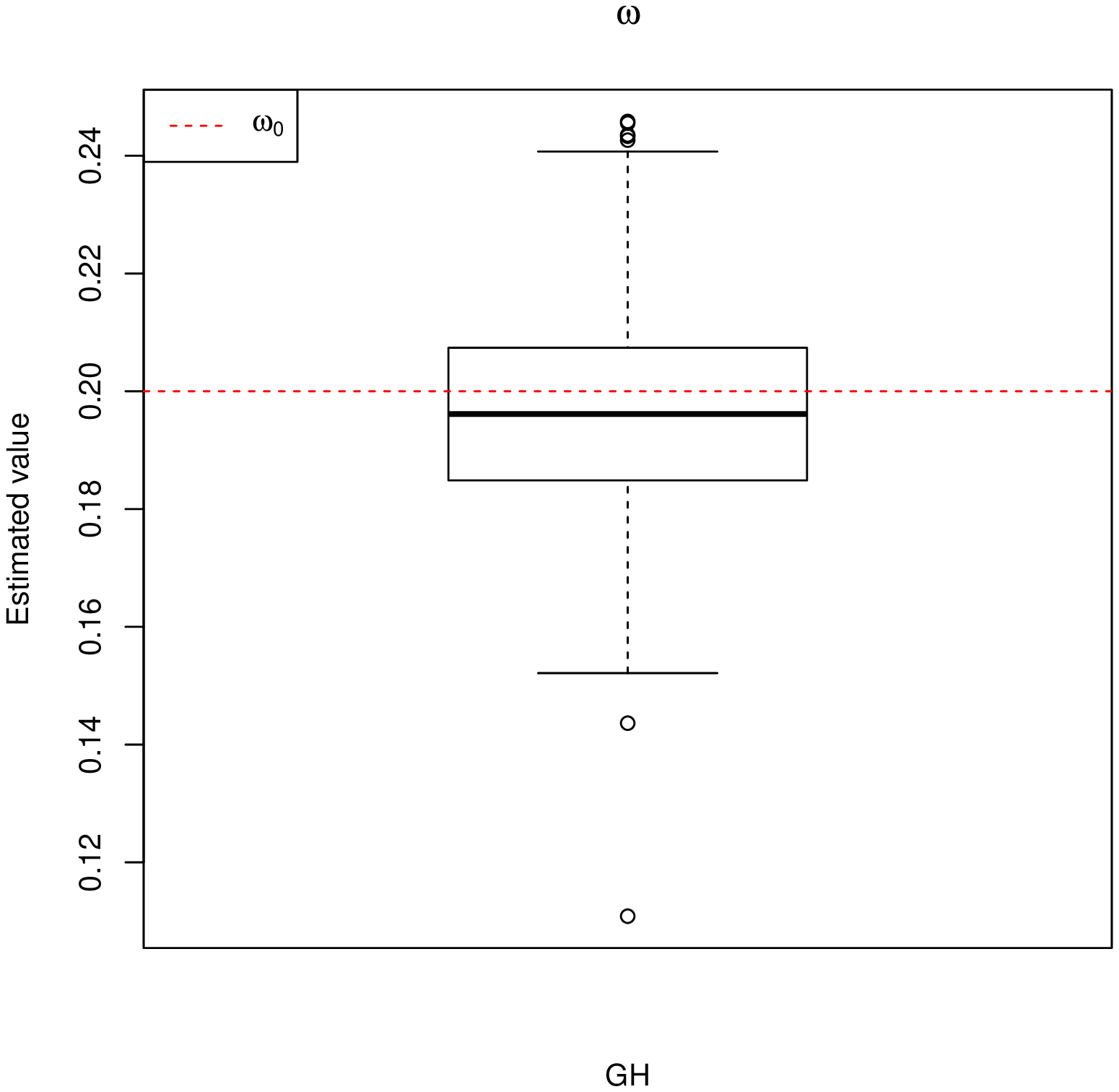}
	\end{minipage}%
	\begin{minipage}[t]{5.6cm}
		\includegraphics[angle=0,width=5.5cm]{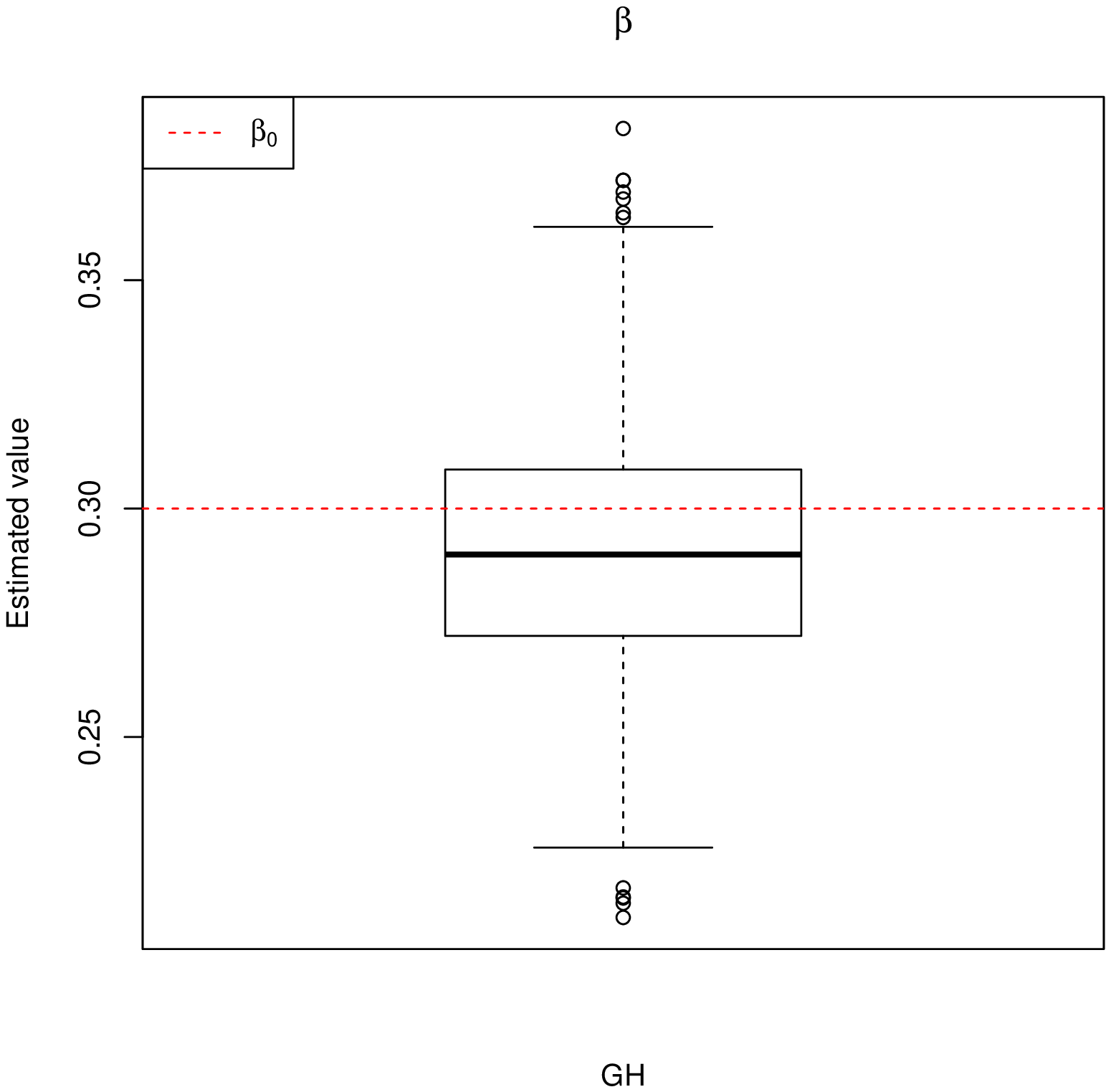}
	\end{minipage}
	
	\begin{minipage}[t]{5.6cm}
		\includegraphics[angle=0,width=5.5cm]{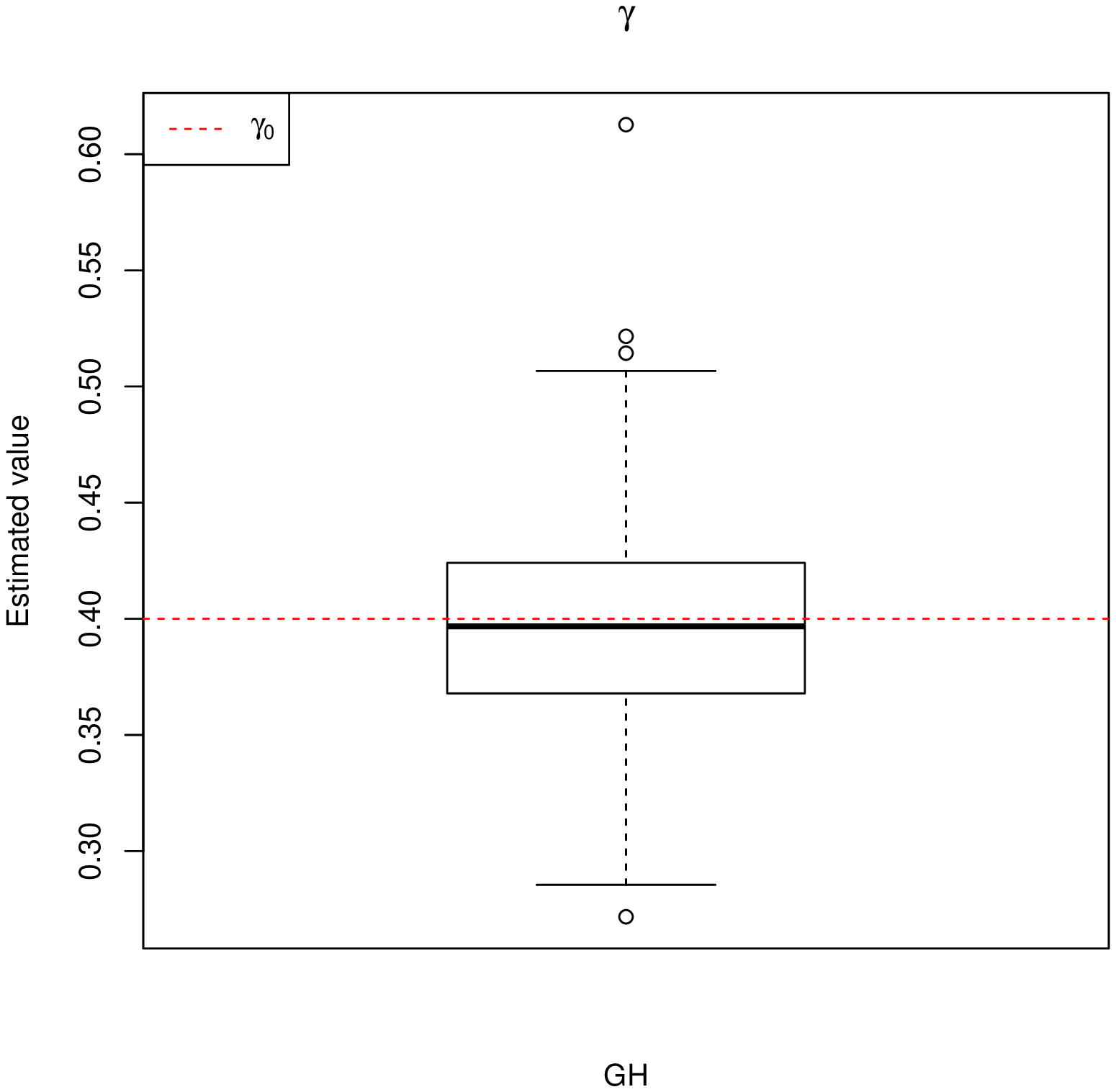}
	\end{minipage}
	\begin{minipage}[t]{5.6cm}
		\includegraphics[angle=0,width=5.5cm]{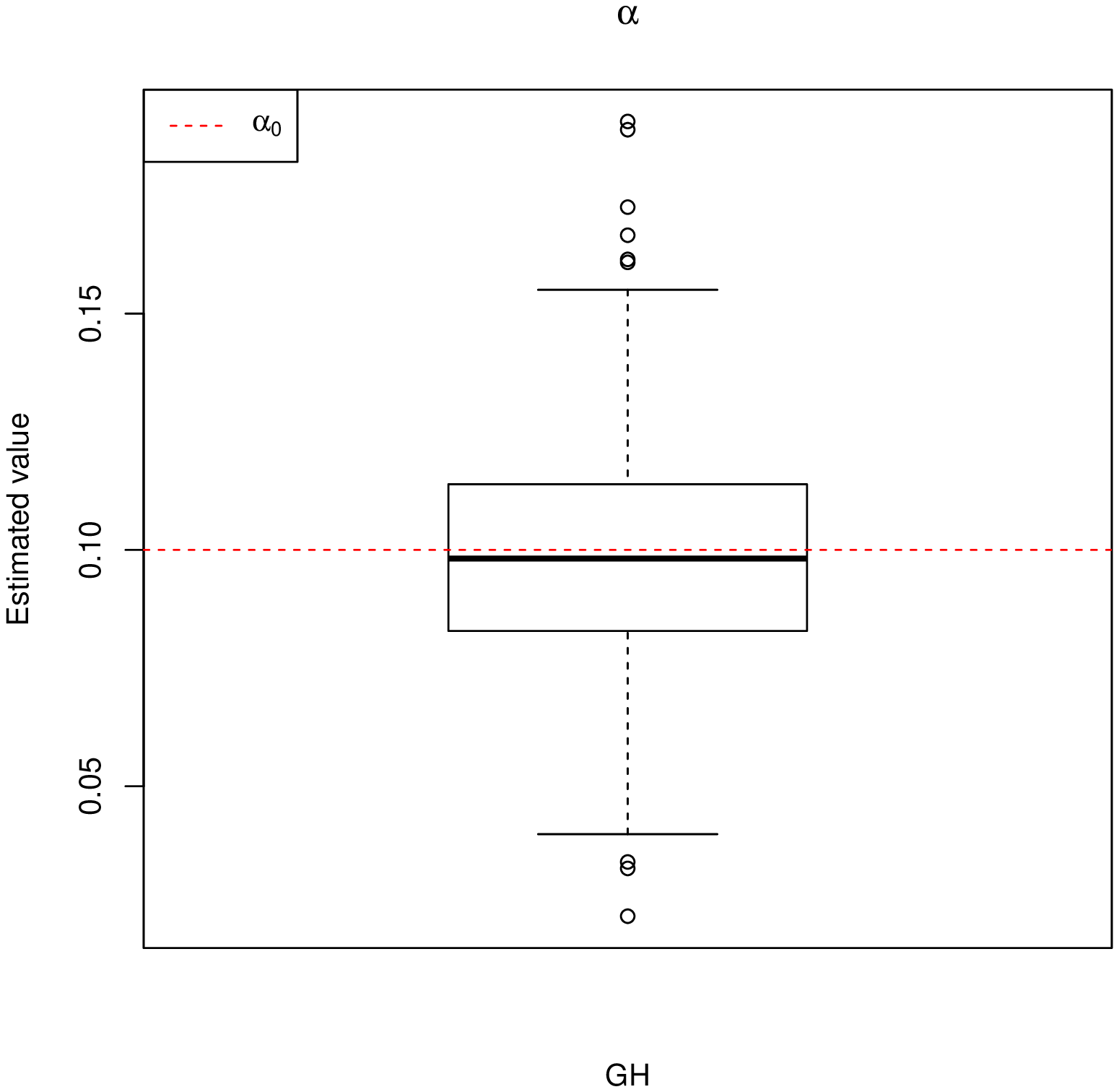}
	\end{minipage}
	\caption{Boxplots of $\tilde{\theta}^{GH}$ for estimating $\theta_{0}=(\omega_{0},\beta_{0},\gamma_{0},\alpha_{0})$ base on simulated data, where GH represents $\tilde{\theta}^{GH}$.}
	\label{fig_boxplot}
\end{figure}
\subsection{Prediction Performance under different theoretical volatility models}
In this subsection, we simulate the sample data from two theoretical volatility models, Heston model and Jump-diffusion model. Under each theoretical volatility model, we simulate high-frequency time interval for 10 seconds, and we report the out-of-sample prediction performances of GQARCH-It\^{o} model for  $\frac{1}{j}, j=1,2,3,4,5,6$ of daily volatility.

\subsubsection{Heston stochastic volatility model}
In this Monte Carlo experiment, we use as the data generating process the stochastic volatility model of Heston(1993) for the instantaneous variance
\begin{align}
\begin{split}
dS(t)=&rS(t)dt+\sqrt{V(t)}S(t)dW_{1}(t)\\
dV(t)=&(a-bV(t))dt+\gamma\sqrt{V(t)}dW_{2}(t).
\end{split}
\end{align}
parameter $a,b,\gamma$ is positive, $W_{1}(t), W_{2}(t)$ are Brownian motions, and $\rho $ is the correlation coefficient between $W_{1}(t)$ and $W_{2}(t)$.
We set parameters $a_{0}=0.01$, $b_{0}=0.001$, $\gamma_{0}=0.075$, $\rho_{0}=-0.8$, and $r_{0}=0.02$. We further set $n=101, S_0=50, V_0=0.05$. We simulated the high-frequency data with 10 seconds interval, and $\frac{1}{2}$, $\frac{1}{3}$, $\frac{1}{4}$, $\frac{1}{5}$ and $\frac{1}{6}$ of daily volatility forecasting results are also presented in the studies.
All simulation is based 1000 repetitions, and the first 100 days are used for in-sample, and the 101th day is saved for the out-of-sample forecast of daily volatility.
We define the following four criteria to evaluate the forecasting error, which are mean absolute error (MAE), mean square error (MSE), adjusted mean absolute percentage error (AMAPE) and logarithmic loss (LL),
\begin{align*}
&MAE=\frac{1}{N}\sum_{i=1}^{N}|RV_{i}-F_i|,\quad &MSE=\frac{1}{N}\sum_{i=1}^{N}(RV_{i}-F_i)^{2},\\
&AMAPE=\frac{1}{N}\sum_{i=1}^{N}\left|\frac{F_i-RV_{i}}{F_i+RV_{i}}\right|,\quad &LL=\frac{1}{N}\sum_{i=1}^{N}(\log(F_i)-\log(RV_{i}))^{2},
\end{align*}
where the realized volatility $RV_i$ is considered as the best estimation of the real integrated volatility in day $i$, $F_i$ is the volatility prediction in day $i$.
All results are presented in Table \ref{table_error_heston 10 seconds 1}.

According to Table \ref{table_error_heston 10 seconds 1}, we may get some interesting findings. First, our proposed model could predict the $\frac{1}{2}$, $\frac{1}{3}$, $\frac{1}{4}$, $\frac{1}{5}$ and $\frac{1}{6}$ of daily volatility. This is very important in financial market since that a large number of securities' practitioners would like to know the volatility performances for next half day, one-third days, one-fourth days, one-fifth days and even one-sixth days rather than future daily volatility. Second, the $\frac{1}{2}$, $\frac{1}{3}$, $\frac{1}{4}$, $\frac{1}{5}$ and $\frac{1}{6}$ of daily volatility prediction results have better performances than future daily volatility via the values of MAE, MSE, APAME  and LL.

\subsubsection{Jump-diffusion model}
The ture price of the security is assumed to obey the following Jump-diffusion model,
\begin{align}
\begin{split}
dS(t)=&rS(t)dt+\sqrt{V(t)}S(t)dW_{1}(t)\\
dV(t)=&(a-bV(t))dt+\gamma\sqrt{V(t)}dW_{2}(t)+dJ_{t},
\end{split}
\end{align}
where $J_{t}=\sum_{i=1}^{N_{{t}}}U_{i}$ is a compound Poisson process, $\{N_t\}$ is a Poisson process with intensity $\lambda$, and random variables $\{U_{i}\}$ are independent,  following the same distribution  $N(0,\sigma_{J}^{2})$. Besides choosing the same parameter values as in Heston model, we further set $\lambda=1$ and $\sigma_J=0.01$. The similar forecasting results  are represented in Table \ref{table_error_heston 10 seconds 2}.

\section{EMPIRICAL STUDY}
%\subsection{The data}

In this section, we try to illustrate the volatility predictor power with trading data second-by-second for CRSP total market index, and  data source is from Wharton Research Data Service (WRDS). The period is from January 2, 2018 to December 31, 2018. The number of analyzed high-frequency data is 5803200, and the number of low-frequency datais 248. All high-frequency prices are transformed into log prices $\log (P_{t_{i,j}})$, $t_{i,j}=i-1+j/m$, $i=1,\ldots,n$, $j=1,\ldots,m$ with $n=228$, $m=23400$.

First, we divide the data into in-sample data and out-sample data. The in-sample period starts from January 2, 2018 to August 31, 2018, which contains 3931200 high-frequency prices and 168 days. The out-of-sample period starts  from September 4, 2018 to December 31, 2018, which contains 187200 high-frequency prices and 80 days. We would explore the volatility forecasting for $\frac{1}{j}, j=1,2,3,4,5,6$ of daily volatility in out-of-sample period. To illustrate the prediction power of the GQARCH-It\^{o} model, we also compute the forecasts of GARCH-It\^{o} model.

Because our proposed GQARCH-It\^{o} model explains the volatility asymmetry, its volatility prediction power is stronger than the one of GARCH-It\^{o} model via comparison Table \ref{table_error 1} and Table \ref{table_error 2}.

\section{CONCLUSION}
In this paper, we introduce a novel GQARCH-It\^{o} model, which could explain volatility asymmetry. Model parameters in the GQARCH-It\^{o} model are estimated by maximizing a quasi-likelihood function. In simulation study and empirical study, we show that GQARCH-It\^{o} model could predict $\frac{1}{j}, j=1,2,3,4,5,6$ of daily volatility prediction. More importantly, the proposed GQARCH-It\^{o} model has stronger forecasting power than GARCH-It\^{o} model.

The proposed GQARCH-It\^{o} model can be also extended in several other directions. First, option data is also another important information for volatility prediction. The new model would have better performance if it consists option data. Second, the parameters of GQARCH-It\^{o} model has the same convergence rate as in GARCH-It\^{o} model, therefore, the quasi-likelihood function should be improved in estimating the model¡¯s parameters. Machine learning is the scientific study of algorithms and statistical models that computer systems use to effectively perform a specific task without using explicit instructions, relying on patterns and inference instead. It seems that some more efficient methodologies could be proposed combining vast amounts of data and machine learning.

\section*{ACKNOWLEDGEMENTS} 
Cui's work was partially supported by \textit{National Natural Science Foundation of China (71671106),  and Zhou's work was partially supported by the State Key Program of National Natural Science Foundation of China (71931004), the State Key Program in the Major Research Plan of National Natural Science Foundation of China (91546202)}.

\section*{DATA AVAILABILITY STATEMENT} 
The data that support the findings of this study are available in [Wharton Research Data Services] at [https://wrds-www.wharton.upenn.edu]. Please refer to the supporting information for details.

\section*{SUPPORTING INFORMATION} 

The data are downloaded from CRSP total market index(January 2, 2018 to December 31, 2018) in Wharton Research Data Services, and  the details may be founded online in the supporting information tab for this article.

%\bibliography{myReferences}

%\section*{REFERENCES}

\newpage
%\section*{TABLES}
\begin{table}[H]
	\caption{$\frac{1}{j}, j=1,2,3,4,5,6$ of daily volatility prediction performance for GQARCH-It\^{o} model under Heston model, the time interval is 10 seconds. }
	\begin{tabular}{ccccccc}
		\hline
		\hline
		&MAE& MSE &  AMAPE &   LL&\\
		\hline
		QGARCH-It\^{o}(10 seconds, 1 day)&{\bf1.723e-04}&{\bf6.887e-08 }& {\bf 0.147}  & {\bf0.147}  \\
		\hline
		QGARCH-It\^{o}(10 seconds,  $\frac{1}{2}$ day)&6.601e-05&3.266e-08& 0.084 & 0.054 \\
		QGARCH-It\^{o}(10 seconds,  $\frac{1}{3}$ day)&4.350e-05&3.661e-09& 0.103 & 0.070 \\
		QGARCH-It\^{o}(10 seconds,  $\frac{1}{4}$ day)&3.465e-05&8.543e-09& 0.107 & 0.071 \\
		QGARCH-It\^{o}(10 seconds,  $\frac{1}{5}$ day)&4.260e-05&4.471e-08& 0.119 & 0.135 \\
		QGARCH-It\^{o}(10 seconds,  $\frac{1}{6}$ day)&2.831e-05&2.611e-09& 0.128 & 0.104 \\
		\hline
		\hline
	\end{tabular}
	\label{table_error_heston 10 seconds 1}
\end{table}
\begin{table}[H]
	\caption{$\frac{1}{j}, j=1,2,3,4,5,6$ of daily volatility prediction performance for GQARCH-It\^{o} model under Heston jump model, the time interval is 10 seconds. }
	\begin{tabular}{ccccccc}
		\hline
		\hline
		&MAE& MSE &  AMAPE &   LL&\\
		\hline
		QGARCH-It\^{o}(10 seconds, 1 day)&{\bf2.104e-04}&{\bf1.153e-07 }& {\bf 0.153}  & {\bf0.158}  \\
		\hline
		QGARCH-It\^{o}(10 seconds,  $\frac{1}{2}$ day)&8.763e-05&3.030e-08& 0.106 & 0.087 \\
		QGARCH-It\^{o}(10 seconds,  $\frac{1}{3}$ day)&3.917e-05&2.500e-08& 0.088 & 0.058 \\
		QGARCH-It\^{o}(10 seconds,  $\frac{1}{4}$ day)&2.787e-05&2.266e-08& 0.089 & 0.068 \\
		QGARCH-It\^{o}(10 seconds,  $\frac{1}{5}$ day)&2.174e-05&1.610e-09& 0.094 & 0.053 \\
		QGARCH-It\^{o}(10 seconds,  $\frac{1}{6}$ day)&2.246e-05&1.479e-08& 0.096 & 0.071 \\
		\hline
		\hline
	\end{tabular}
	\label{table_error_heston 10 seconds 2}
\end{table}
\begin{table}[H]
	\caption{$\frac{1}{j}, j=1,2,3,4,5,6$ of daily volatility prediction performance for CRSP based on GQARCH-It\^{o} model, the time interval is 1 second. }
	\begin{tabular}{ccccccc}
		\hline
		\hline
		&MAE& MSE &  AMAPE &   LL&\\
		\hline
		QGARCH-It\^{o}(1 day)&{\bf3.976e-05}&{\bf3.501e-09 }& {\bf 0.232}  & {\bf0.344}  \\
		\hline
		GQARCH-It\^{o}($\frac{1}{2}$ day)&7.294e-07&4.257e-11& 0.005 & 0.008 \\
		GQARCH-It\^{o}($\frac{1}{3}$ day)&5.163e-07&2.132e-11& 0.005 & 0.008 \\
		GQARCH-It\^{o}($\frac{1}{4}$ day)&5.550e-07&2.464e-11& 0.006 & 0.016 \\
		GQARCH-It\^{o}($\frac{1}{5}$ day)&4.713e-07&1.777e-11& 0.007 & 0.017 \\
		GQARCH-It\^{o}($\frac{1}{6}$ day)&4.563e-07&1.665e-11& 0.007 & 0.019 \\
		\hline
		\hline
	\end{tabular}
	\label{table_error 1}
\end{table}
\begin{table}[H]
	\caption{$\frac{1}{j}, j=1,2,3,4,5,6$ of daily volatility prediction performance for CRSP based on GARCH-It\^{o} model, the time interval is 1 second. }
	\begin{tabular}{ccccccc}
		\hline
		\hline
		&MAE& MSE &  AMAPE &   LL&\\
		\hline
		GARCH-It\^{o}(1 day)&{\bf4.466e-05}&{\bf4.601e-09 }& {\bf 0.255}  & {\bf0.422}  \\
		GARCH-It\^{o}($\frac{1}{2}$ day)&1.006e-06&8.105e-11& 0.006 & 0.012 \\
		GARCH-It\^{o}($\frac{1}{3}$ day)&8.523e-07&5.811e-11& 0.006 & 0.015 \\
		GARCH-It\^{o}($\frac{1}{4}$ day)&8.975e-07&6.444e-11& 0.008 & 0.027 \\
		GARCH-It\^{o}($\frac{1}{5}$ day)&5.900e-07&2.785e-11& 0.007 & 0.022 \\
		GARCH-It\^{o}($\frac{1}{6}$ day)&4.566e-07&1.666e-11& 0.007 & 0.019 \\
		\hline
		\hline
	\end{tabular}
	\label{table_error 2}
\end{table}

\section*{APPENDIX}

\subsection*{A1. Proof of Proposition 1 }

\begin{proof}
	(a) By It\^{o} Lemma, we have
	\begin{align*}
	d \beta\left(\int_{n-1}^{t}\sigma_{s}dB_{s}\right)^2 =2\beta\left(\int_{n-1}^{t}\sigma_{s}dB_{s}\right)\sigma_{t}dB_{t}+\beta\sigma_{t}^{2}dt.
	\end{align*}
	Then, for $n-\frac{1}{j}<t<n$, we further have
	\begin{align*}
	R(k)\equiv & \int_{n-\frac{1}{j}}^{n}\frac{(n-t)^{k}}{k!}\sigma_{t}^{2}dt\\
	= & \frac{ \omega +(\gamma +j(k+2)-1) \sigma_{n-\frac{1}{j}}^{2}}{(k+2)!}\left(\frac{1}{j}\right)^{k+2}+\beta R(k+1)+2 \beta
	\int_{n-\frac{1}{j}}^{n}\frac{(n-t)^{k+1}}{(k+1)!}\int_{n-\frac{1}{j}}^{t}\sigma_{s}dB_{s}\sigma_{t}dB_{t}\\
	+&\alpha \int_{n-\frac{1}{j}}^{n}\frac{(n-t)^{k+1}}{(k+1)!}\sigma_{t}dB_{t}.
	\end{align*}
	
	By the iteration of $R(k)$ , we can obtain
	\begin{align*}
	&\int_{n-\frac{1}{j}}^{n}\sigma_{t}^{2}dt\\
	=&\sum_{k=0}^{\infty} \frac{\beta^{k} \left[\omega +(\gamma +j(k+2)-1) \sigma_{n-\frac{1}{j}}^{2}\right]}{(k+2)!}\left(\frac{1}{j}\right)^{k+2}+ \sum_{k=0}^{\infty}2
	\int_{n-\frac{1}{j}}^{n}\frac{\left[\beta(n-t)\right]^{k+1}}{(k+1)!}\int_{n-\frac{1}{j}}^{t}\sigma_{s}dB_{s}\sigma_{t}dB_{t}\\
	+&\sum_{k=0}^{\infty}\int_{n-\frac{1}{j}}^{n}\frac{\alpha\left(n-t\right)^{k+1}}{(k+1)!}\beta^{k}\sigma_{t}dB_{t}\\
	=& g_{n}(\theta)+D_{n},
	\end{align*}
	where
	\begin{align}\label{eq:h_nj}
	\begin{split}
	g_{n}(\theta)=&\sum_{k=0}^{\infty} \frac{\beta^{k} \left[\omega +(\gamma +j(k+2)-1) \sigma_{n-\frac{1}{j}}^{2}\right]}{(k+2)!}\left(\frac{1}{j}\right)^{k+2}\\
	=&\omega_{j}^{g}+\frac{\gamma+j-1}{j} g_{n-\frac{1}{j}}(\theta)+\beta_{j}^{g}Z_{n-\frac{1}{j}}^{2}+\alpha_{j}^{g}Z_{n-\frac{1}{j}},
	\end{split}
	\end{align}
	and
	\begin{align*}
	\omega_{j}^{g}=&\beta^{-1}\left(e^{\frac{\beta}{j}}-1\right)\frac{\omega}{j},\quad \beta_{j}^{g}=\beta^{-1}(\gamma-1)\left(e^{\frac{\beta}{j}}-1-\frac{\beta}{j}\right)+e^{\frac{\beta}{j}}-1,\\
	\alpha_{j}^{g}=&\alpha\left(\beta^{-2}(\gamma-1)\left(e^{\frac{\beta}{j}}-1-\frac{\beta}{j}\right)+\beta^{-1}\left(e^{\frac{\beta}{j}}-1\right)\right),
	\end{align*}
	By Taylor expansion of $e^{(n-t)\beta}$ , $D_{n}$ can be written as
	\begin{align*}
	D_{n}=2\int_{n-\frac{1}{j}}^{n}\left(e^{(n-t)\beta}-1\right)\int_{n-\frac{1}{j}}^{t}\sigma_{s}dB_{s}\sigma_{t}dB_{t}+\frac{\alpha}{\beta}\int_{n-\frac{1}{j}}^{n}(e^{\left(n-t\right)\beta}-1)\sigma_{t}dB_{t},
	\end{align*}
	As the integrand of $D_n$ is predictable, $D_n$ is a martingale difference.
	
	(b) It is an immediate consequence of $E [D_n|\mathcal{F}_{n-1}] = 0$.
\end{proof}
\subsection*{A1. Proof of Theorem 1 }
Let
\begin{align*}
&\tilde{L}_{n,m}^{GH}(\theta)=-\frac{1}{2n}\sum_{i=1}^{n}\log(g_{i}(\theta))-\frac{1}{2n}\sum_{i=1}^{n}\frac{RV_{i}}{g_{i}(\theta)}\equiv-\frac{1}{2n}\sum_{i=1}^{n}\tilde{l}_{i}^{GH}(\theta)
\quad and \quad \tilde{\psi}_{n,m}^{GH}(\theta)=\frac{\partial\tilde{L}_{n,m}^{GH}(\theta) }{\partial\theta},\\
&\tilde{L}_{n}^{GH}(\theta)=-\frac{1}{2n}\sum_{i=1}^{n}\log(g_{i}(\theta))-\frac{1}{2n}\sum_{i=1}^{n}\frac{\int_{i-1}^{i}\sigma_{t}^{2}dt}{g_{i}(\theta)} \quad and \quad
\tilde{\psi}_{n}^{GH}(\theta)=\frac{\partial\tilde{L}_{n}^{GH}(\theta) }{\partial\theta},\\
&{L}_{n}^{GH}(\theta)=-\frac{1}{2n}\sum_{i=1}^{n}\log(g_{i}(\theta))-\frac{1}{2n}\sum_{i=1}^{n}\frac{g_{i}(\theta_{0})}{g_{i}(\theta)} \quad and \quad
\psi_{n}^{GH}(\theta)=\frac{\partial{L}_{n}^{GH}(\theta) }{\partial\theta},
\end{align*}
Let $(\omega_l^g,\omega_u^g)$, $(\beta_l^g,\beta_u^g)$, $(\gamma_l^g,\gamma_u^g)$ and $(\alpha_l^g,\alpha_u^g)$ be the lower bound and the upper bound of $\omega^g$, $\beta^g$, $\gamma$ and $\alpha^g$. To ease notations, we denote derivatives of any function $g$ at $x_{0}$ by
\begin{align*}
\frac{\partial g(x_{0})}{\partial x}=\frac{\partial g(x)}{\partial x}\bigg|_{x=x_{0}}.
\end{align*}
We first provide two useful lemmas.
	
\begin{lemma}
	Under Assumption 1 (a), (b), for the GQARCH-It\^{o} model, we have
	\\
	\\
	(a). there exists a neighborhood $B(\theta_{0})$ of $\theta_{0}$ such that for any $p\geq1$, $\sup\limits_{i\in \mathbb{N}}\left\| \sup\limits_{\theta\in
		B(\Theta_{0})}\frac{g_{i}(\theta_{0})}{g_{i}(\theta)}\right\|_{L_{p}}<\infty$ and $B(\theta_{0})\subset\Theta$.
	\\
	\\
	(b). for any $ p\geq 1$, $\sup\limits_{i\in \mathbb{N}}\left\| \sup\limits_{\theta\in\Theta}\frac{\partial g_{i}(\theta)}{\partial\theta_{j}}\right\|_{L_{p}}\leq C$,
	$\sup\limits_{i\in \mathbb{N}}$
	$\left\| \sup\limits_{\theta\in\Theta}\frac{\partial^{2} g_{i}(\theta)}{\partial\theta_{j}\partial \theta_{k}}\right\|_{L_{p}}\leq C
	$, and $\sup\limits_{i\in \mathbb{N}}\left\| \sup\limits_{\theta\in\Theta}\frac{\partial^{3} g_{i}(\theta)}{\partial\theta_{j}\partial \theta_{k}\partial
		\theta_{v}}\right\|_{L_{p}}\leq C$ for any $j,k,v \in \{1,2,3,4\}$, where $\theta=(\theta_{1},\theta_{2},\theta_{3},\theta_{4})=(\omega,\beta,\gamma,\alpha)$;
\end{lemma}
\begin{proof}
	(a) By the iteration of $g_{i}(\theta)$, we have
	\begin{align*}
	g_{i}(\theta)=\sum_{k=0}^{i-3}(\omega_{1}^{g}+\beta_{1}^{g}Z_{i-k-1}^{2}+\alpha_{1}^{g} Z_{i-k-1})\gamma^{k}+\gamma^{i-2}g_{2}
	(\theta),
	\end{align*}
	where
	\begin{align}
	g_{2}(\theta)=\beta^{-2}(e^{\beta}-1-\beta)\omega+[\beta^{-2}(\gamma-1)(e^{\beta}-1-\beta)+\beta^{-1}(e^{\beta}-1)]\sigma_{1}^{2}< \infty.
	\end{align}
	Choose $s\in[0,1]$ such that $\sup_{i \in \mathbb{N}} E(Z_{i}^{2ps})<\infty$.  Then, similar to the proof of Lemma 2(d) of \cite{KW:2016}, it is easy to obtain the following result,
	\begin{align*}
	\sup\limits_{i\in \mathbb{N}}\left\| \sup\limits_{\theta\in
		B(\Theta_{0})}\frac{g_{i}(\theta_{0})}{g_{i}(\theta)}\right\|_{L_{p}}<\infty.
	\end{align*}
	
	(b)  We first prove that the first order derivatives are finite. We have
	\begin{align*}
	\frac{\partial g_{i}(\theta)}{\partial \alpha_{1}}= &\sum_{k=0}^{i-3}\beta_{1}^{-1}\beta_{1}^{g}\gamma^{k}Z_{i-k-1}+\gamma^{i-2}\frac{\partial g_{2}
		(\theta)}{\partial\alpha}\\
	\leq & \sum_{k=0}^{i-3}\beta_{1}^{-1}\beta_{1}^{g}\gamma^{k}Z_{i-k-1}+C.
	\end{align*}
	By noticing that $x/(x+1)\leq x^{s}$ for any $x\geq 0$ and any $s\in [0,1]$, we can show
	\begin{align*}
	g_{i}(\theta)^{-1}\frac{\partial g_{i}(\theta)}{\partial \alpha_{1}}
	\leq & \sum_{k=0}^{i-3}\frac{\beta_{1}^{-1}\beta_{1}^{g}Z_{i-k-1}}{\omega_{1}^{g}+\beta_{1}^{g}Z_{i-k-1}^{2}+\alpha_{1}^{g} Z_{i-k-1}}+C\\
	\leq & C\sum_{k=0}^{i-3} (\alpha_{1}^g Z_{i-k-2})^{s}+C.
	\end{align*}
	Under Assumption 1 (b), we have
	\begin{align*}
	\sup\limits_{i\in \mathbb{N}}\left\|\sup\limits_{\theta\in \Theta}\frac{\partial g_{i}(\theta)}{\partial \alpha_{1}}\right\|_{L_{p}}\leq C.
	\end{align*}
	Applying the same argument, we can also prove that
	\begin{align*}
	\sup\limits_{i\in \mathbb{N}}\left\|\sup\limits_{\theta\in \Theta}\frac{\partial g_{i}(\theta)}{\partial \omega_{1}}\right\|_{L_{p}}\leq C,\quad \sup\limits_{i\in \mathbb{N}}\left\|\sup\limits_{\theta\in \Theta}\frac{\partial g_{i}(\theta)}{\partial \beta_{1}}\right\|_{L_{p}}\leq C.
	\end{align*}
	and
	\begin{align*}
	\sup\limits_{i\in \mathbb{N}}\left\|\sup\limits_{\theta\in \Theta}\frac{\partial g_{i}(\theta)}{\partial\gamma}\right\|_{L_{p}}\leq C.
	\end{align*}
	
	Finally, we can similarly show the boundedness for the second order, and third order derivatives.
\end{proof}
	
\begin{lemma}
	Under Assumption (a), (b), (d), (f) and (g), we have
	\begin{align*}
	\sup_{\theta\in
		\Theta } \left|\tilde{L}_{n,m}^{GH} (\theta)-L_n^{GH}(\theta)\right|= O_p(m^{-1/4})+o_p(1).
	\end{align*}
\end{lemma}

\begin{proof}
	The differences of integrated volatilities between the GQARCH-It\^{o} model and the GARCH-It\^{o} model of \cite{KW:2016} is the martingale difference term. Furthermore, the asymmetry information only acts the variance of parameters, have no effect on the convergence rate of parameters. Therefore, similar to the proof of Lemma 3 of \cite{KW:2016}, we can obtain the result.
\end{proof}

\begin{lemma}
	Under Assumption (a), (b) and (h), we have
	\\
	\\
	(a). there exists a neighborhood $B(\theta_{0})$ of $\theta$ such that $\sup\limits_{i\in \mathbb{N}}\left\|\sup_{\theta\in
		B(\theta_{0})}\frac{\partial^{3}\tilde{l}_{i}^{GH}(\theta)}{\partial\theta_{j}\partial\theta_{k}\partial\theta_{v}}\right\|_{L_{1}}<\infty$ for any $j,k,v\in
	\left\{1,2,3,4\right\}$, where $\theta=(\theta_{1},\theta_{2}, \theta_{3}, \theta_{4})=(\omega,\beta,\gamma,\alpha)$.
	\\
	\\
	(b). $-\nabla \psi_{n}^{GH}(\theta_{0})$ is a positive definite matrix for $n\geq4$.
\end{lemma}

\begin{proof}
	(a) For any $j,k,v\in \left\{1,2,3,4\right\}$, we can obtain
	\begin{align*}
	\frac{\partial^{3}\tilde{l}_{i}^{GH}(\theta)}{\partial\theta_{j}\partial\theta_{k}\partial\theta_{v}}=&
	\left\{1-\frac{RV_{i}}{g_{i}(\theta)}\right\}\left\{\frac{1}{g_{i}(\theta)}\frac{\partial^{3}g_{i}(\theta)}{\partial\theta_{j}\partial\theta_{k}\partial\theta_{v}}\right\}\\
	&+ \left\{2\frac{RV_{i}}{g_{i}(\theta)}-1\right\}\left\{\frac{1}{g_{i}(\theta)}\frac{\partial
		g_{i}(\theta)}{\partial\theta_{j}}\right\}\left\{\frac{1}{g_{i}(\theta)}\frac{\partial^{2} g_{i}(\theta)}{\partial\theta_{k}\partial\theta_{v}}\right\}\\
	&+ \left\{2\frac{RV_{i}}{g_{i}(\theta)}-1\right\}\left\{\frac{1}{g_{i}(\theta)}\frac{\partial g_{i}(\theta)}{\partial
		\theta_{k}}\right\}\left\{\frac{1}{g_{i}(\theta)}\frac{\partial^{2} g_{i}(\theta)}{\partial\theta_{j}\partial\theta_{v}}\right\}\\
	&+\left\{2\frac{RV_{i}}{g_{i}(\theta)}-1\right\}\left\{\frac{1}{g_{i}(\theta)}\frac{\partial
		g_{i}(\theta)}{\partial\theta_{v}}\right\}\left\{\frac{1}{g_{i}(\theta)}\frac{\partial^{2} g_{i}(\theta)}{\partial\theta_{j}\partial\theta_{k}}\right\}\\
	&+\left\{2-6\frac{RV_{i}}{g_{i}(\theta)}\right\}\left\{\frac{1}{g_{i}(\theta)}\frac{\partial g_{i}(\theta)}{\partial\theta_{j}}\right\}
	\left\{\frac{1}{g_{i}(\theta)}\frac{\partial g_{i}(\theta)}{\partial\theta_{k}}\right\}
	\left\{\frac{1}{g_{i}(\theta)}\frac{\partial g_{i}(\theta)}{\partial\theta_{v}}\right\}
	\end{align*}
	By Assumption 1 (h), we can get
	\begin{align*}
	E[RV_{i}|\mathcal{F}_{i-1}]\leq CE\left[\int_{i-1}^{i}\sigma_{t}^{2}dt|\mathcal{F}_{i-1}\right]+C \quad a.s..
	\end{align*}
	Then, by Lemma 1, the tower property and H\"{o}lder's inequality, we have
	\begin{align*}
	&E\left[\sup\limits_{\theta\in
		B(\theta_{0})}\left|\frac{RV_{i}}{g_{i}(\theta)}\left\{\frac{1}{g_{i}(\theta)}\frac{\partial^{3}g_{i}(\theta)}{\partial\theta_{j}\partial\theta_{k}\partial\theta_{v}}\right\}\right|\right]\\
	\leq & C E\left[\sup\limits_{\theta\in
		B(\theta_{0})}\frac{g_{i}(\theta_{0})}{g_{i}(\theta)}\left|\frac{1}{g_{i}(\theta)}\frac{\partial^{3}g_{i}(\theta)}{\partial\theta_{j}\partial\theta_{k}\partial\theta_{v}}\right|\right]+C\\
	\leq &C \left\|\sup\limits_{\theta\in B(\theta_{0})}\frac{g_{i}(\theta_{0})}{g_{i}(\theta)}\right\|_{L_p} \left\|\sup\limits_{\theta\in
		B(\theta_{0})}\left|\frac{1}{g_{i}(\theta)}\frac{\partial^{3}g_{i}(\theta)}{\partial\theta_{j}\partial\theta_{k}\partial\theta_{v}}\right|\right\|_{L_{q}}
	+C\leq C<\infty
	\end{align*}
	where $1/p+1/q=1,p>1$ and $q>1$. Similarly, we can prove that other terms are also bounded.
	
	(b) It is easy to show that
	\begin{align*}
	-\nabla \psi_{n}^{GH}(\theta_{0})=\frac{1}{2n}\sum_{i=1}^{n}\frac{\partial g_{i}(\theta_{0})}{\partial\theta}\frac{\partial
		g_{i}(\theta_{0})^{T}}{\partial\theta}g_{i}(\theta_{0})^{-2}=\frac{1}{2n}\sum_{i=1}^{n}g_{\theta,i}g_{\theta,i}^{T}
	\end{align*}
	where $g_{\theta,i}=\frac{\partial g_{i}(\theta_{0})}{\partial\theta}g_{i}(\theta_{0})^{-1}$. Suppose that $-\nabla \psi_{n}^{GH}(\theta_{0})$ is not a positive definite
	matrix. Then, there exists $\lambda\neq \mathbf{0}$ such that $\frac{1}{2n}\sum_{i=1}^{n}\lambda ^{T}g_{\theta,i}g_{\theta,i}^{T}\lambda=0$, which further implies
	\begin{align*}
	g_{\theta,i}^{T}\lambda =0 \quad a.s. \quad \mbox{ for all} \quad i=1,\ldots, n.
	\end{align*}
	Since $g_{i}(\theta_{0})$ stays away from zero, we have
	\begin{align*}
	\begin{pmatrix}
	\frac{\partial g_{1}(\theta_{0})}{\partial\omega} & \frac{\partial g_{1}(\theta_{0})}{\partial\beta} & \frac{\partial g_{1}(\theta_{0})}{\partial\gamma} &\frac{\partial
		g_{1}(\theta_{0})}{\partial\alpha}\\
	\frac{\partial g_{2}(\theta_{0})}{\partial\omega} & \frac{\partial g_{2}(\theta_{0})}{\partial\beta}&
	\frac{\partial g_{2}(\theta_{0})}{\partial\gamma}&
	\frac{\partial g_{2}(\theta_{0})}{\partial\alpha}\\
	\vdots & \vdots & \vdots & \vdots\\
	\frac{\partial g_{n}(\theta_{0})}{\partial\omega} & \frac{\partial g_{n}(\theta_{0})}{\partial\beta} & \frac{\partial g_{n}(\theta_{0})}{\partial\gamma}&\frac{\partial g_{n}(\theta_{0})}{\partial\alpha}\\
	\end{pmatrix}
	\lambda=\mathbf{0} \quad a.s.,
	\end{align*}
	where
	\begin{align*}
	&\frac{\partial g_{i+1}(\theta_{0})}{\partial\omega}=\frac{\partial\omega_{0}^{g}}{\partial\omega} +\gamma\frac{\partial g_{i}(\theta_{0})}{\partial\omega},\\
	&\frac{\partial g_{i+1}(\theta_{0})}{\partial\beta}= \frac{\partial\omega_{0}^{g}}{\partial\beta}+\gamma\frac{\partial
		g_{i}(\theta_{0})}{\partial\beta}+\frac{\partial\beta_{0}^{g}}{\partial\beta}Z_{i}^{2}+\frac{\partial\alpha_{0}^{g}}{\partial\beta}Z_{i},\\
	&\frac{\partial g_{i+1}(\theta_{0})}{\partial\gamma}=g_{i}(\theta_{0})+\gamma\frac{\partial
		g_{i}(\theta_{0})}{\partial\gamma}+\frac{\partial\beta_{0}^{g}}{\partial\gamma}Z_{i}^{2}+\frac{\partial\alpha_{0}^{g}}{\partial\gamma}Z_{i},\\
	& \frac{\partial g_{i+1}(\theta_{0})}{\partial\alpha}= \gamma\frac{\partial
		g_{i}(\theta_{0})}{\partial\alpha}+\frac{\partial\alpha_{0}^{g}}{\partial\alpha}Z_{i},
	\end{align*}
	and  $\frac{\partial\omega_{0}^{g}}{\partial\omega} =\beta_{0}^{-1}(e^{\beta_{0}}-1)$, $\frac{\partial\omega_{0}^{g}}{\partial\beta}
	=\beta_{0}^{-2}(1-e^{\beta_{0}})\omega_{0}+\beta_{0}^{-1}e^{\beta_{0}}\omega_{0}$,
	$\frac{\partial\beta_{0}^{g}}{\partial\beta}=(\gamma_{0}-1)
	(\beta_{0}^{-1}e^{\beta_{0}}-\beta_{0}^{-2}e^{\beta_{0}}+\beta_{0}^{-2})+e^{\beta_{0}}$, $\frac{\partial\alpha_{0}^{g}}{\partial\beta}=
	(\beta_{0}^{-2}\gamma_{0}(e^{\beta_{0}}-1)-2\beta_{0}^{-3}(\gamma_{0}-1)(e^{\beta_{0}}-\beta_{0}-1)-2\beta_{0}^{-2}(e^{\beta_{0}}-1)+\beta_{0}^{-1}e^{\beta_{0}})\alpha$,
	$\frac{\partial\beta_{0}^{g}}{\partial\gamma}=\beta_{0}^{-1}\alpha_{0}(e^{\beta_{0}}-1-\beta_{0})$, $\frac{\partial\alpha_{0}^{g}}{\partial\gamma}=\alpha_{0}\beta_0^{-2}(e^{\beta_0}-1-\beta_0)$, $\frac{\partial\alpha_{0}^{g}}{\partial\alpha}=\beta_0^{-2}(\gamma_{0}-1)(e^{\beta_0}-1-\beta_0)+\beta_{0}^{-1}(e^{\beta_{0}}-1)$. Since $Z_{i}$'s and $O_i$'s are nondegenerate, the
	matrix on the left hand side is of full rank a.s., which implies $\lambda=\mathbf{0}$ a.s. Thus, it is a contradiction to the initial assumption.
\end{proof}

\textbf{Proof of Theorem 1 }

(a) According to the definition of
$L_{n}^{GH}(\theta)$, we have
\begin{align*}
\max \limits_{\theta\in\Theta}~L_{n}^{GH}(\theta) \leq
-\frac{1}{2n}\sum_{i=1}^{n}\min\limits_{\theta_{i}\in\Theta}~\log(g_{i}(\theta_{i}))+\frac{g_{i}(\theta_{0})}{g_{i}(\theta_{i})}.
\end{align*}
If $\theta_{0i}$ satisfies $g_{i}(\theta_{0i})=g_{i}(\theta_{0})$, $\theta_{0i}$ is the minimizer of $\log(g_{i}(\theta_{i}))+\frac{g_{i}(\theta_{0})}{g_{i}(\theta_{i})}$.Thus, if $\theta^{*}\in\Theta$ satisfies $g_{i}(\theta^{*})=g_{i}(\theta_{0})$ for all $i=1,2,\ldots,n$, $\theta^{*}$ is the maximizer of $L_{n}^{GH}(\theta)$. Next, we show that $\theta^{*}$ must be equal $\theta_{0}$ a.s. Since
\begin{align*}
g_{i}(\theta)=\omega_{1}^{g}+\gamma g_{i-1}(\theta)+\beta_{1}^{g}Z_{i-1}^{2}+\alpha_{1}^{g}Z_{i-1},
\end{align*}
both $\theta^{*}$ and $\theta_{0}$ satisfy the following equation,
\begin{align*}
\begin{pmatrix}
1 & g_{1}(\theta_{0}) & Z_{1}^{2} & Z_{1}\\
1 & g_{2}(\theta_{0}) & Z_{2}^{2} & Z_{2}\\
\vdots & \vdots & \vdots & \vdots \\
1 & g_{n-1}(\theta_{0}) & Z_{n-1}^{2} & Z_{n-1}\\
\end{pmatrix}
\begin{pmatrix}
\omega^{*g}-\omega_{0}^{g}\\
\gamma^{*}-\gamma_{0}\\
\beta^{*g}-\beta_{0}^{g}\\
\alpha^{*g}-\alpha_{0}^{g}\\
\end{pmatrix}
\equiv M
\begin{pmatrix}
\omega^{*g}-\omega_{0}^{g}\\
\gamma^{*}-\gamma_{0}\\
\beta^{*g}-\beta_{0}^{g}\\
\alpha^{*g}-\alpha_{0}^{g}\\
\end{pmatrix}
=\mathbf{0} \quad \mbox {a.s.},
\end{align*}
where $\omega^{*g}=\beta^{*-1}(e^{\beta^{*}}-1)\omega^{*}$, $\beta^{*
	g}=\beta^{*-1}(\gamma^{*}-1)(e^{\beta^{*}}-1-\beta^{*})+e^{\beta^{*}}-1$,
and $\alpha^{*g}=\alpha^{*}(\beta^{*-2}(\gamma-1)(e^{\beta^{*}}-1-\beta^{*})+\beta^{*-1}(e^{\beta^{*}}-1))$. Since $Z_{i}$'s  are nondegenerate,
$M$ is of full rank, which implies that $M^{T}M$ is invertible and
\begin{align*}
\begin{pmatrix}
\omega^{*g}-\omega_{0}^{g}\\
\gamma^{*}-\gamma_{0}\\
\beta^{*g}-\beta_{0}^{g}\\
\alpha^{*g}-\alpha_{0}^{g}\\
\end{pmatrix}=\mathbf{0} \quad \mbox{a.s.}
\end{align*}
For given $\gamma$, $\beta^{g}$ is strictly increasing function with respect to $\beta$ and for given $\beta$, $\alpha^g$ is strictly increasing function with respect to $\alpha$ and $\beta$. Then, we have $\theta^{*}=\theta_{0}$, i.e., there is a unique maximizer of $L_{n}^{GH}(\theta)$. Then, since $L_{n}^{GH}(\theta)$ is a continuous function, for any $\varepsilon>0$, there is a constant $c$, such that
\begin{align*}
L_{n}^{GH}(\theta_{0})-\max \limits_{\theta\in\Theta: \left\| \theta-\theta_{0}\right\|_{max}\geq \varepsilon}L_{n}^{GH}(\theta)> c \quad \mbox{a.s.}
\end{align*}
With the help of Theorem 1 in \cite{X:2010} and Lemma 2, we can derive the conclusion.

(b) Applying Taylor expansion and Rolle mean value theorem, we have
\begin{align*}
\tilde{\psi}_{n,m}^{GH}(\tilde{\theta}^{GH})-\tilde{\psi}_{n,m}^{GH}(\theta_{0})=-\tilde{\psi}_{n,m}^{GH}(\theta_{0})=\nabla\tilde{\psi}_{n,m}^{GH}(\theta^{*})(\tilde{\theta}^{GH}-\theta_{0})
\end{align*}
where $ \theta^{*}$ is between $\theta_{0}$ and $\tilde{\theta}^{GH}$. According to Lemma 3 (b), $-\nabla \psi_{n}^{GH}(\theta_{0})$ is a positive matrix. If we can show $-\nabla
\tilde{\psi}_{n,m}^{GH}(\theta^{*}) \xrightarrow{p} -\nabla \psi_{n}^{GH}(\theta_{0})$, the convergence rate of $\tilde{\theta}^{GH} - \theta_{0}$ is the same as that of
$\tilde{\psi}_{n,m}^{GH}(\theta_{0})$.

We first show that
\begin{align*}
\tilde{\psi}_{n,m}^{GH}(\theta_{0})=O_{p}(m^{-1/4})+O_{p}(n^{-1/2}).
\end{align*}
For any $j\in \{1,2,3,4\}$, by Lemma 1 and H\"{o}lder's inequality, we have
\begin{align}
\nonumber \left\|
\tilde{\psi}_{n,m}^{GH}(\theta_{0})-\tilde{\psi}_{n}^{GH}(\theta_{0})
\right\|_{L_1}=& \left\|
\frac{1}{2n}\sum_{i=1}^{n}\frac{\partial g_{i}(\theta_{0})}{\partial\theta_{j}}g_{i}(\theta_{0})^{-2}\left(RV_{i}-\int_{i-1}^{i}\sigma_{s}^{2}dt\right)
\right\|_{L_{1}}\\
\nonumber \leq &C\frac{1}{n}\sum_{i-1}^{n}\left\|
\frac{\partial g_{i}(\theta_{0})}{\partial\theta_{j}}g_{i}(\theta_{0})^{-1}\right\|_{L_{q}}\left\|
RV_{i}-\int_{i=1}^{i}\sigma_{s}^{2}dt\right\|_{L_{p}}\\
\label{eq:order_3}\leq & C m^{-1/4}
\end{align}
where $1<p\leq1+\delta$ and $1/p+1/q=1$ and the last inequality is due to Assumption 1 (g).
Then, we have
\begin{align*}
\tilde{\psi}_{n,m}^{GH}(\theta_{0})=\frac{1}{2n}\sum_{i=1}^{n}\frac{\partial g_{i}(\theta_{0})}{\partial\theta}g_{i}(\theta_{0})^{-1}\frac{D_{i}}{g_{i}(\theta_{0})}+O_{p}(m^{-1/4})
\end{align*}
Applying It\^{o}'s lemma and It\^{o} isometry, we have for any $j\in \{1,2,3,4\}$,
\begin{align}
\nonumber &E\left[\left(\frac{1}{2n}\sum_{i=1}^{n}\frac{\partial
	g_{i}(\theta_{0})}{\partial\theta_{j}}g_{i}(\theta_{0})^{-1}\frac{D_{i}}{g_{i}(\theta_{0})}\right)^{2}\right]\\
\nonumber =&\frac{1}{4n^{2}}\sum_{i=1}^{n}E\left[\left(\frac{\partial
	g_{i}(\theta_{0})}{\partial\theta_{j}}\right)^{2}g_{i}(\theta_{0})^{-2}\frac{D_{i}^{2}}{g_{i}^{2}(\theta_{0})}\right]\\
\nonumber =& \frac{1}{4n^{2}}\sum_{i=1}^{n}E\left[\left(\frac{\partial
	g_{i}(\theta_{0})}{\partial\theta_{j}}\right)^{2}g_{i}(\theta_{0})^{-2}\frac{E[D_{i}^{2}|\mathcal{F}_{i-1}]}{g_{i}^{2}(\theta_{0})}\right]\\
\label{eq:order_1}\leq & C \frac{1}{n^{2}}\sum_{i=1}^{n}E\left[\left(\frac{\partial
	g_{i}(\theta_{0})}{\partial\theta_{j}}\right)^{2}g_{i}(\theta_{0})^{-2}\frac{E[Z_{i}^{4}|\mathcal{F}_{i-1}]}{g_{i}^{2}(\theta_{0})}\right].
\end{align}
According to Assumption (c) and Lemma 2 (b), we know that (\ref{eq:order_1}) is of order $n^{-1}$. Thus, we further have
\begin{align*}
\tilde{\psi}_{n,m}^{GH}(\theta_{0})=O_{p}(m^{-1/4})+O_{p}(n^{-1/2}).
\end{align*}

Then, we show that
\begin{align*}
\left\|\nabla
\tilde{\psi}_{n,m}^{GH}(\theta^{*})-\nabla\psi_{n}^{GH}(\theta_{0})\right\|_{\max}=o_p(1).\end{align*}
By the triangular inequality, we have
\begin{align}
\left\|
\nabla
\tilde{\psi}_{n,m}^{GH}(\theta^{*})-\nabla\psi_{n}^{GH}(\theta_{0})\right\|_{\max}\leq & \left\|
\nabla
\tilde{\psi}_{n,m}^{GH}(\theta^{*})-\nabla\tilde{\psi}_{n,m}^{GH}(\theta_{0})\right\|_{\max}+\left\|
\nabla
\tilde{\psi}_{n,m}^{GH}(\theta_{0})-\nabla\psi_{n}^{GH}(\theta_{0})\right\|_{\max}. \label{eq:order_2}
\end{align}
For the first term on the right side of (\ref{eq:order_2}), noticing Theorem 1 (a) and Lemma 3 (a), we have
\begin{align*}
\left\|
\nabla
\tilde{\psi}_{n,m}^{GH}(\theta^{*})-\nabla\tilde{\psi}_{n,m}^{GH}(\theta_{0})\right\|_{\max}\leq &\frac{C}{n}\sum_{i=1}^{n} \max_{j,k,v\in
	{(1,2,3,4)}^{3}}\sup_{\theta\in B(\theta_{0})}\left|\frac{\partial^{3}
	\tilde{l}_{i}^{GH}(\theta)}{\partial\theta_{j}\partial\theta_{k}\partial\theta_{v}}\right|\left\|\theta^{*}-\theta_{0}\right\|_{\max} = o_{p}(1).
\end{align*}
For the second term on the right side of (\ref{eq:order_2}), similar to the proof of (\ref{eq:order_3}), by H\"{o}lder's inequality and Lemma 1, we have
\begin{align*}
\left\|
\nabla
\tilde{\psi}_{n,m}^{GH}(\theta_{0})-\nabla\tilde{\psi}_{n}^{GH}(\theta_{0})\right\|_{\max}
=& O_{p}(m^{-1/4})
\end{align*}
Therefore, we derive that
\begin{align*}
\nabla
\tilde{\psi}_{n,m}^{GH}(\theta_{0})=\nabla\tilde{\psi}_{n}^{GH}(\theta_{0})
+ O_{p}(m^{1/4})=\nabla\psi_{n}^{GH}(\theta_{0})+\zeta_{n}+ O_{p}(m^{-1/4}) \quad \mbox {a.s.},
\end{align*}
where $\zeta_{n}=\frac{1}{2n}\sum_{i=1}^{n}\left(\frac{\partial ^{2}
	g_{i}(\theta_{0})}{\partial\theta\partial\theta^{T}}g_{i}(\theta_{0})^{-1}\frac{-D_{i}}{g_{i}(\theta_{0})}+\frac{\partial
	g_{i}(\theta_{0})}{\partial\theta}\left(\frac{\partial g_{i}(\theta_{0})}{\partial\theta}\right)^{T}g_{i}(\theta_{0})^{-2}\frac{2D_{i}}{g_{i}(\theta_{0})}\right)$ is a martingale. Similar to the proof of (\ref{eq:order_1}), we can show $\left\|\zeta_{n}\right\|_{\max}=O_{p}(n^{-1/2})$.
Thus, we obtain
\begin{align*}
\left\|\nabla
\tilde{\psi}_{n,m}^{GH}(\theta^{*})-\nabla\psi_{n}^{GH}(\theta_{0})\right\|_{\max}=o_p(1).\end{align*}

Finally, with the help of the above two results, we have $\left\|\tilde{\theta}^{GH}-\theta_{0}\right\|_{\max} = O_{p}\left(m^{-1/4}+n^{-1/2}\right).$

\subsection*{A3. Proof of Theorem 2}

\begin{proof}
	As
	\begin{align*}
	\tilde{\psi}_{n,m}^{GH}(\tilde{\theta}^{GH})-\tilde{\psi}_{n,m}^{GH}(\theta_{0})=-\tilde{\psi}_{n,m}^{GH}(\theta_{0})=\nabla\tilde{\psi}_{n,m}^{GH}(\theta^{*})(\tilde{\theta}^{GH}-\theta_{0}),
	\end{align*}
	where $ \theta^{*}$ is between $\theta_{0}$ and $\tilde{\theta}^{GH}$, we have
	\begin{align*}
	\sqrt{n}(\tilde{\theta}^{GH}-\theta_{0})=&-\sqrt{n}\left(\nabla\tilde{\psi}_{n,m}^{GH}(\theta^{*})\right)^{-1}\tilde{\psi}_{n,m}^{GH}(\theta_{0})\\
	=&-\sqrt{n}\left(\nabla\tilde{\psi}_{n,m}^{GH}(\theta^{*})\right)^{-1}\left(\tilde{\psi}_{n}^{GH}(\theta_{0})+O_{p}(m^{-1/4})\right)\\
	=&-\sqrt{n}\left(\nabla\psi_{n}^{GH}(\theta_{0})+o_{p}(1)\right)^{-1}\tilde{\psi}_{n}^{GH}(\theta_{0})+o_{p}(1),
	\end{align*}
	where the second and third equality is due to the proof of Theorem 1(b).
	
	As $\lambda^T \frac{\partial
		g_{i}(\theta_{0})}{\partial\theta}g_{i}(\theta_{0})^{-1}\frac{D_{i}}{g_{i}(\theta_{0})}$ is stationary and ergodic, we have
	\begin{align*}
	\sqrt{n}\tilde{\psi}_{n}^{GH}(\theta_{0})=\sqrt{n}\frac{1}{2n}\sum_{i=1}^{n}\frac{\partial
		g_{i}(\theta_{0})}{\partial\theta}g_{i}(\theta_{0})^{-1}\frac{D_{i}}{g_{i}(\theta_{0})}\xrightarrow {d} N(0,A^{GH}),
	\end{align*}
	by using Cram\'{e}r-Wold device and the martingale central limit theorem.
	
	On the other hand, by the proof of Theorem 1(b),
	\begin{align*}
	-\nabla\psi_{n}^{GH}(\theta_{0})=&\frac{1}{2n}\sum_{i=1}^{n}\left(\frac{1}{g_{i}^{2}(\theta_{0})}\frac{\partial g_{i}(\theta_{0})}{\partial\theta}\left(\frac{\partial g_{i}(\theta_{0})}{\partial\theta}\right)^{T}\right)\rightarrow  B \quad \mbox{in probability}.
	\end{align*}
	Therefore,
	\begin{align*}
	\sqrt{n}(\tilde{\theta}^{GH}-\theta_{0})\xrightarrow {d}& N(0,B^{-1}A^{GH}B^{-1}).
	\end{align*}
\end{proof}

\end{document}